\documentclass{article}

\usepackage{microtype}
\usepackage{graphicx}
\usepackage{subcaption}
\usepackage{booktabs} 
\usepackage{enumitem}
\usepackage{hyperref}

\usepackage[preprint]{icml2026}

\usepackage{amsmath}
\usepackage{amssymb}
\usepackage{mathtools}
\usepackage{amsthm}
\usepackage{multirow}
\usepackage{makecell}
\usepackage{amsthm}
\usepackage{graphicx}
\usepackage{textcomp}
\usepackage{xcolor}
\usepackage{booktabs}
\usepackage{float} 

\usepackage[capitalize,noabbrev]{cleveref}

\theoremstyle{plain}
\newtheorem{theorem}{Theorem}[section]

\theoremstyle{definition}

\theoremstyle{remark}

\usepackage[textsize=tiny]{todonotes}

\begin{document}

\twocolumn[
  \icmltitle{IEMAS: An Incentive-Efficiency Routing Framework for Open Agentic Web Ecosystems}



  \icmlsetsymbol{equal}{*}

  \begin{icmlauthorlist}
    \icmlauthor{Hongze Liu}{equal,yyy}
    \icmlauthor{Chang Guo}{equal,yyy}
    \icmlauthor{Yingzeng Li}{yyy}
    \icmlauthor{Mengru Wang}{yyy}
    \icmlauthor{Jiong Lou}{yyy}
    \icmlauthor{Shijing Yuan}{yyy}
    \icmlauthor{Hefeng Zhou}{yyy}
    \icmlauthor{Chentao Wu}{yyy}
    \icmlauthor{Jie Li}{yyy}
  \end{icmlauthorlist}

  \icmlaffiliation{yyy}{School of Computer Science, Shanghai Jiao Tong University, China}

  \icmlcorrespondingauthor{Jie Li}{lijiecs@sjtu.edu.cn}

  \icmlkeywords{Machine Learning, ICML}

  \vskip 0.3in
]



\printAffiliationsAndNotice{}  

\begin{abstract}
  The transition to open, distributed Multi-Agent Systems (MAS) promises scalable intelligence but introduces a non-trivial tension: maximizing global efficiency requires cooperative, resource-aware scheduling, yet autonomous agents may be self-interested and cannot be managed by a centralized controller. Prior approaches fall short in two key areas: they typically focus on single-query  routing, neglecting long-term resource reuse (e.g., KV-caching) and the complexities of system-level many-to-many matching; furthermore, they rely on generic incentive mechanisms that ignore the distinct characteristics of LLM inference. To bridge this gap, we propose \textbf{IEMAS} (\textbf{I}ncentive-\textbf{E}fficiency Mechanism for \textbf{M}ulti-\textbf{A}gent \textbf{S}ystems), a  framework that aligns economic incentives with system performance. IEMAS integrates a probabilistic predictive model to estimate Quality of Service (QoS) under uncertainty, which feeds into a VCG-based bipartite matching mechanism. This design guarantees truthful capability reporting and social optimality while explicitly leveraging KV cache-affinity to minimize computational redundancy. We implement IEMAS on top of vLLM and evaluate it via extensive simulations. Results demonstrate that our incentive-efficiency co-design reducing average service cost by \textbf{35\%} and end-to-end latency by up to \textbf{2.9$\times$} compared to baselines.
\end{abstract}

\section{Introduction}
Large Language Models have enabled autonomous agents with strong reasoning and planning abilities, motivating recent work on multi-agent LLM systems where complementary agents collaborate to solve complex tasks \cite{guo2024large}. In practical deployments, such agents are often hosted by heterogeneous entities and distributed across networked infrastructures \cite{guo2025betaweb}. Emerging paradigms, such as BetaWeb \cite{guo2025betaweb}, Agentic Web\cite{yang2025agentic}, DAWN \cite{aminiranjbar2025dawn}, and the Internet of Agents (IoA) \cite{cheninternet, wang2025internet}, envision open, decentralized ecosystems where agents dynamically coordinate over the network to serve diverse client requests via flexible workflows.

However, coordinating LLM agents in open, distributed ecosystems creates three interdependent challenges~\cite{wang2025internet}. \textbf{P1. Scheduling and routing:} Matching concurrent client requests to heterogeneous, capacity-limited agents to maximize efficiency while preserving KV-prefix locality and meeting service constraints. \textbf{P2. Incentives:} Independent providers with private costs may misreport capabilities or selectively serve requests unless economically motivated. \textbf{P3. Communication scalability:} Broadcasting agent details or context histories between all participants incurs prohibitive latency and cost, rendering full transparency infeasible.

Most prior work on request routing (\textbf{P1}) relies on centralized controllers or trusted environments, assumptions that break down at web scale \cite{kwon2023vllm}. While some studies address planning for specific multi-agent workflows \cite{yang2025agentnet}, they typically assume a \emph{one-to-many} dispatch structure. In contrast, open networks exhibit a \emph{many-to-many} structure where multiple heterogeneous requests must be jointly matched to distributed agents. Consequently, routing strategies optimal for single-query dispatch become suboptimal in multi-agent scenarios. Crucially, while techniques like key--value (KV) caching yield order-of-magnitude speedups for related contexts, naïve load balancing destroys cache locality, significantly increasing latency and computational cost \cite{kwon2023vllm,microsoft2023mii}.

Beyond scheduling inefficiencies, the absence of verifiable mechanisms (\textbf{P2}) introduces strategic risks: agents may misreport capabilities or behave opportunistically. Without incentive alignment, even carefully engineered schedulers fail. While mechanism design in distributed systems has shown promise in steering strategic actors, existing proposals using game theory \cite{wang2024social}, auctions \cite{you2024privacy}, contract theory \cite{ye2024optimizing}, and persuasion \cite{zhong2025hybrid} generally treat agents as homogeneous resources. They fail to incorporate constraints intrinsic to LLM services, such as unpredictable generation lengths \cite{sun2024llumnix,chen2025kairos}, skill specialization \cite{yang2025agentnet}, and the critical value of KV-cache reuse \cite{pankvflow}. Furthermore, time-to-first-token (TTFT) effects directly shape perceived latency \cite{chen2025kairos}. Ignoring these factors compromises the incentive structures required to elicit truthful, high-quality participation.

To design an \emph{organic} LLM-MAS, mechanisms must explicitly incorporate these specific factors. Such integration enables allocation rules that preserve cache locality and balance short-term throughput with long-term cooperation. While auction theory provides a foundation for incentive design under asymmetric information \cite{qiu2022applications}, adapting it to LLM agent inference is non-trivial due to KV-cache dependencies and stochastic performance \cite{sun2024llumnix,chen2025kairos}. Moreover, in open environments, agent information (\textbf{P3}) is often incomplete, limiting the applicability of centralized optimization and motivating the need for distributed, incentive-aware mechanisms.

We address these gaps by proposing \textbf{IEMAS} (\textbf{I}ncentive-\textbf{E}fficiency Mechanism for \textbf{M}ulti-\textbf{A}gent \textbf{S}ystems), a distributed routing framework that matches client requests to suitable agents while jointly aligning economic incentives and exploiting computational resources such as KV caching. IEMAS introduces a lightweight proxy layer, where each serving node deploys a gateway for query management and cache coordination, reducing communication overhead and enabling enforceable pricing without revealing proprietary internals. By leveraging KV-prefix caching, multi-turn queries reuse prior computation, while a cache-aware predictive model captures quality-of-service (QoS) signals, including performance, latency, partial cache states, and costs. Based on these signals, clients and agents participate in proxy-level auctions, and a Min-Cost Max-Flow (MCMF) mechanism determines the final request-to-agent routing by jointly considering bidding outcomes, capability alignment, and dynamic workload conditions. We implement IEMAS on \textbf{vLLM}, demonstrating stable agent utilities and improved long-term social welfare.

To the best of our knowledge, this work is the first to investigate an incentive–efficiency co-design routing framework in agentic web ecosystems. The main contributions are summarized as follows:
\begin{itemize}[leftmargin=*,topsep=0em]
\setlength{\parskip}{0em}
\setlength{\itemsep}{0em}
\item We formalize the client--agent interaction model for distributed LLM agent services with agent specialization.
\item We develop a cache-aware predictive scheme that estimates capability, expected latency, and cost metrics.
\item We design an incentive-compatible auction mechanism that promotes truthful reporting and social welfare.
\item We implement the proposed framework on vLLM and perform extensive experiments.
\end{itemize}
The related code is available at: \url{https://github.com/PACHAKUTlQ/IMMAS}.
\section{Related Work}
\label{sec:related_work}

Our work bridges distributed multi-agent systems, efficient LLM serving, and algorithmic game theory. We review these domains and position IEMAS within the emerging landscape of the open Agentic Web.

\noindent
\textbf{LLM-MAS and the Agentic Web:} 
The paradigm of LLM-based Multi-Agent Systems (LLM-MAS) is shifting from closed, monolithic frameworks to open, decentralized networks. Early systems like AutoGEN \cite{wu2023autogen}, MetaGPT \cite{hong2024metagpt}, and CAMEL \cite{li2023camel} demonstrated collaborative problem-solving but assumed centralized control and cooperative agents. Recent research envisions an \textbf{Agentic Web} or \textbf{Internet of Agents (IoA)}, where heterogeneous agents dynamically discover and coordinate across the internet. Protocols such as IoA \cite{chen2024internet} and OpenAgents \cite{xie2023openagents} establish the connectivity layer for this vision, while infrastructures like BetaWeb \cite{guo2025betaweb} and DAWN \cite{aminiranjbar2025dawn} focus on trust and decentralization. However, while these frameworks solve connectivity, they largely overlook the critical economic and system-level inefficiencies of distributing heavy inference workloads. Unlike IEMAS, they lack mechanisms to optimize the high data-movement costs of context transfer or prevent free-riding in trustless environments.

\noindent
\textbf{LLM Routing and Cache-Aware Scheduling:}
Efficient routing is paramount for memory-bound Transformer inference. Optimizations like vLLM \cite{kwon2023vllm} and Orca \cite{yu2022orca} maximize throughput via non-contiguous memory management, while SGLang \cite{zheng2023sglang} and Prompt Cache \cite{gim2024prompt} utilize RadixAttention to achieve order-of-magnitude latency reductions via KV-prefix reuse. However, existing routing strategies (e.g., S-LoRA \cite{sheng2023slora}, Splitwise \cite{patel2023splitwise}) are designed for centralized clusters with full state visibility. In open environments, naïve load balancing scatters semantically related queries, destroying cache locality. IEMAS addresses this gap by introducing \textit{incentive-compatible cache-aware routing}, ensuring requests are routed based on the economic valuation of cached states rather than simple availability.

\noindent
\textbf{Mechanism Design for LLM Ecosystems:}
As agents become autonomous economic actors, mechanism design is required to ensure truthful reporting. While auction theory is well-established for cloud resource allocation, standard models do not account for the stochastic, "stateful" nature (context cache) of LLM inference. Recent works on AI marketplaces \cite{bansal2025magentic, yang2025agent} and truthful capability revelation \cite{liu2025cast} address strategic risks but often decouple the mechanism from system performance, treating compute as a generic commodity. IEMAS bridges this gap by embedding cache-affinity scores directly into a VCG-based auction \cite{vickrey1961counterspeculation, clarke1971multipart}. This ensures the market selects agents that are not only capable but also architecturally positioned (via active KV-caches) to execute tasks efficiently.
\section{System Formulation}
\label{sec:system}

\begin{figure}[tp]
\setlength{\belowcaptionskip}{-0.3cm} 
\setlength{\abovecaptionskip}{0cm}
    \centering
    \includegraphics[width=1\linewidth]{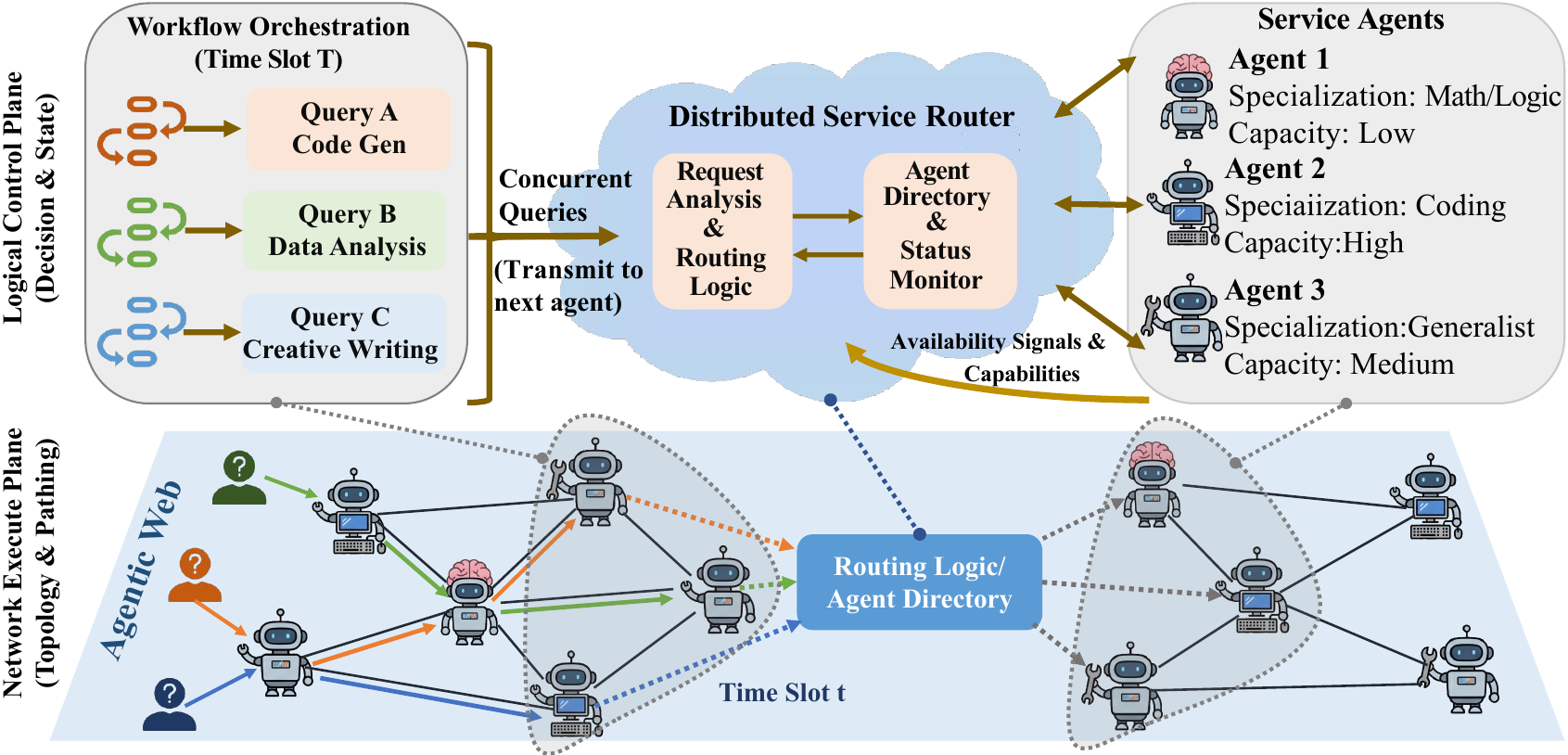}
   \caption{The Illustration of  Agentic Web Routing.}
    \label{fig:agentnetwork}
\end{figure}

We consider a distributed LLM serving system consisting of a set of clients and a set of autonomous LLM agents. Let $\mathcal{C}=\{1,\dots,N\}$ denote the set of clients (prior agents from workflows) and $\mathcal{S}=\{1,\dots,M\}$ denote the set of serving agents. At a time slot, every client $j \in \mathcal{C}$ concurrently submits a task characterized by a semantic context $T_j$.
 Each agent $i\in\mathcal{S}$ is described by a model profile $(S_i,K_i)$, where $S_i$ denotes the model scale (e.g., parameter size or compute footprint) and $K_i$ captures its domain specialization. Agent $i$ has a finite service capacity $B_i$, representing the maximum number of concurrent tasks it can process, and maintains a local key--value (KV) cache that stores prefixes from previously processed contexts.

To capture the benefit of context reuse, we model the semantic overlap between task $j$ and the cached state of agent $i$ by a score $o_{ij}\in[0,1]$. A higher $o_{ij}$ indicates greater KV reuse, which can significantly reduce inference latency and effective computation cost. Accordingly, the service cost incurred by agent $i$ when executing task $j$ is modeled as $C_i = C_i(T_j,S_i,o_{ij})$, which decreases as cache overlap $o_{ij}$ increases. Each client derives utility from successful task execution, which depends on both output quality and latency. Let $P_j(T_j,S_i,K_i)$ denote the expected performance or quality of serving task $j$ using agent $i$, and let $L_j(T_j,S_i,o_{ij})$ denote the corresponding latency or TTFT. We model the client’s valuation as a weighted combination of these factors:
\vspace{-5pt}
\begin{equation}
\setlength\belowdisplayskip{3pt}
\label{eq:individual_value}
v_j \;=\; \delta\, P_j(T_j,S_i,K_i) \;-\; (1-\delta)\, L_j(T_j,S_i,o_{ij}),
\end{equation}
where $\delta\in[0,1]$ captures the client’s relative preference between output quality and latency.

Task allocation is formulated as a many-to-many matching problem and implemented via a double-auction mechanism. Let $x_{ij}\in\{0,1\}$ indicate whether task $j$ is assigned to agent $i$. Feasible allocations must satisfy
\vspace{-5pt}
\begin{equation}
\setlength\belowdisplayskip{3pt}
\sum_{i\in\mathcal{S}} x_{ij} \le 1,\quad \forall j\in\mathcal{C}, 
\qquad
\sum_{j\in\mathcal{C}} x_{ij} \le B_i,\quad \forall i\in\mathcal{S}.
\end{equation}
Each client has a private valuation $v_j$, while each agent has a private service cost or ask price $c_i$. Given reported bids and asks, the auctioneer determines both the allocation $\{x_{ij}\}$ and the corresponding payments. Ignoring strategic behavior, an efficient allocation maximizes total social welfare:
\vspace{-5pt}
\begin{equation}
\setlength\belowdisplayskip{3pt}
W \;=\; \sum_{i\in\mathcal{S}} \sum_{j\in\mathcal{C}} \bigl(v_j - C_i\bigr)\, x_{ij},
\end{equation}
subject to the capacity and assignment constraints above. In each periodic auction round, clients and agents submit bids and asks, and the auctioneer applies the matching and pricing algorithm described in \S\ref{sec:mechanisms}.

A key challenge in this setting is that, unlike traditional distributed computing systems, service cost, latency, and performance for LLM-based agents are not deterministic functions of task and resource parameters~\cite{sun2024llumnix,chen2025kairos}. Instead, these quantities exhibit substantial variability due to factors such as prompt structure, output length, and cache state. Consequently, prior to running the auction, the platform maintains a predictive model
\vspace{-5pt}
\begin{equation}
\setlength\belowdisplayskip{3pt}
\mathcal{M}:\; (T_j, S_i, K_i, o_{ij}) \longmapsto (C_i, P_j, L_j),\nonumber
\end{equation}
which estimates the distributions or expectations of cost, performance, and latency. Recent work on LLM routing and scheduling suggests that although these metrics are stochastic, they often follow stable distributions conditioned on $(T_j,S_i,K_i,o_{ij})$~\cite{chen2025kairos}. The auction mechanism then operates on these predicted quantities and is designed to satisfy feasibility, as well as additional economic desiderata such as incentive compatibility and individual rationality.

\section{IEMAS}
\label{sec:IceMAS}

\begin{figure*}[t]
    \centering
    \includegraphics[width=1\linewidth]{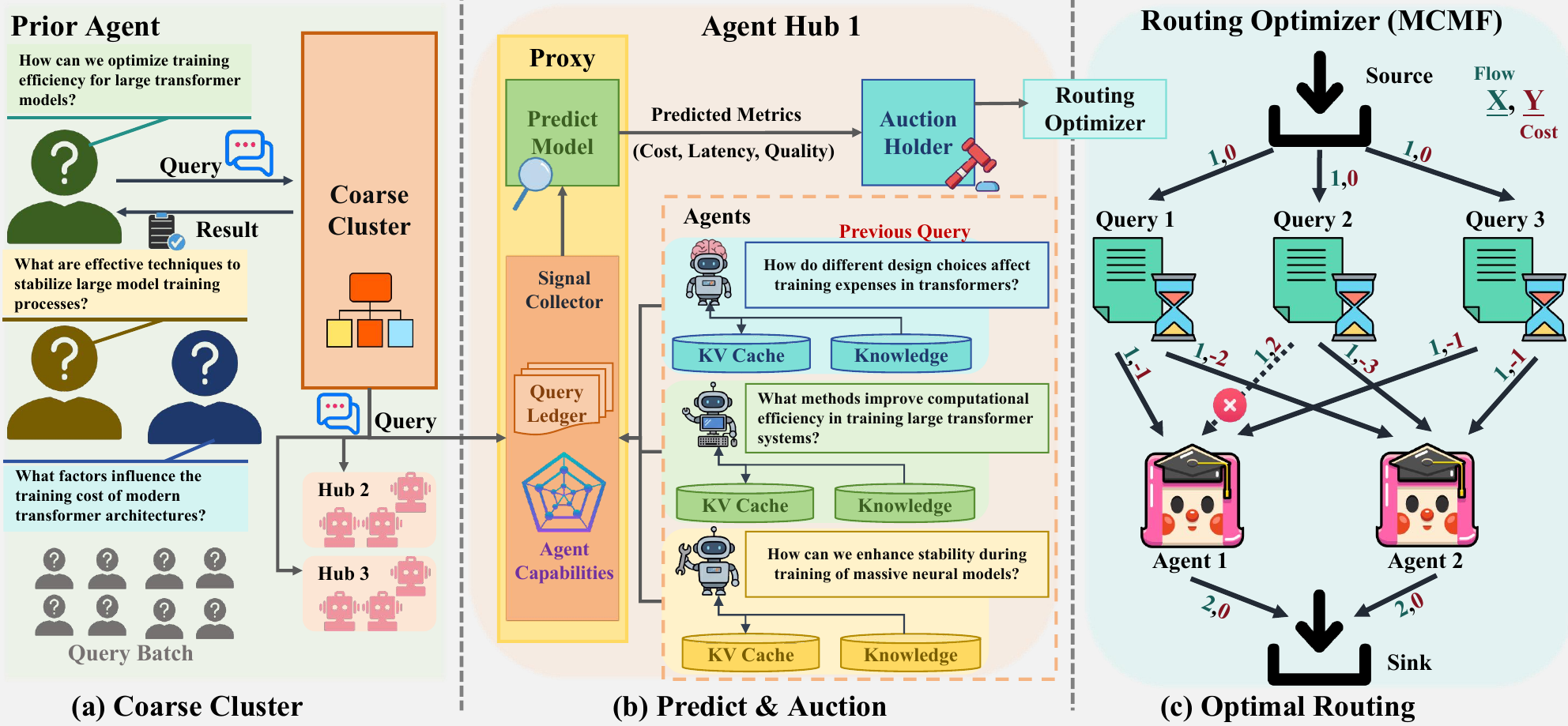}
    \caption{\textbf{IEMAS Overview.} (a) \textbf{Coarse-Grained Clustering:} Incoming web queries are first allocated to specific Agent Hubs via a fast, domain-based clustering mechanism. (b) \textbf{Predictive Auction:} A proxy layer utilizes predictive modeling to generate uncertainty-aware bids/asks and executes an auction to match tasks to agents under capacity constraints. (c) \textbf{Optimization:} The allocation is solved as a Min-Cost Max-Flow (MCMF) problem to maximize social welfare based on truthful bidding.}
    \label{fig:IceMAS}
\end{figure*}

\textbf{IEMAS} is an integrated architecture designed to align economic incentives with system-level efficiency in the open Agentic Web. Unlike prior work on request routing or agent selection that primarily considers scheduling a single task in isolation, large-scale LLM serving in open networks naturally induces a \emph{many-to-many} matching problem between concurrent client requests and heterogeneous agents. In such settings, myopic per-request routing decisions can be socially suboptimal, as they fail to account for capacity constraints, cross-request interactions, and cache locality effects, as empirically demonstrated in \S\ref{sec:experiments}. Following prior work, we introduce a \textbf{proxy hub} as a trusted mediator~\cite{wang2025internet,yang2025agent}, which aggregates information across requests and agents and computes a joint allocation to achieve socially efficient scheduling.

The framework is organized as follows. \S\ref{sec:kv-design} introduces the cache-aware predictive layer that preserves KV-cache locality and derives calibrated QoS signals. \S\ref{sec:mechanisms} presents a VCG-based mechanism that formulates task–agent assignment as a Min-Cost Max-Flow (MCMF) problem to maximize social welfare. \S\ref{sec:theoretical_analysis} analyzes key theoretical properties, including truthfulness, budget balance, and exactness. Finally, \S\ref{sec:hub-design} describes the Proxy Hub architecture, which improves scalability and enforces economic constraints by clustering agents and localizing matching.
Figure~\ref{fig:IceMAS} and Algorithm~\ref{alg:iemas_overview} provide an overview of the complete workflow.

\subsection{Resource-Aware Predictive Modeling}
\label{sec:kv-design}

IEMAS maintains a predictive model, which estimates these metrics(i.e., latency, cost, and quality) for every candidate pair $(i,j)$ to drive the auction described in \S\ref{sec:mechanisms}. The module operates in an online loop: \emph{(1) calculate cache-locality, (2) predict QoS, (3) update via bandit feedback.}

\textbf{Prefix-Locality (KV Reuse Proxy):} To estimate the cache affinity $o_{ij}$, the proxy maintains a \emph{prefix ledger} for each agent $i$. This ledger stores the text $\bar{p}_{i,d}$ of the last executed prompt for a specific dialogue session $d$. Given a new request $j$ belonging to session $d(j)$ with prompt $p_j$, the proxy computes the Longest Common Prefix (LCP) length $l_{ij} = \operatorname{lcp}(p_j, \bar{p}_{i,d(j)})$. The affinity score is defined as:
\vspace{-5pt}
\begin{equation}
\setlength\belowdisplayskip{5pt}
o_{ij} = \frac{l_{ij}}{\max(1, |p_j|)} \in [0, 1].
\label{eq:cache_affinity}
\end{equation}
Because the ledger is agent-specific, switching agents for a multi-turn conversation results in $o_{ij} \approx 0$, correctly capturing the loss of locality.

\paragraph{QoS Online Prediction:}
IEMAS maintains an independent predictor $g_i$ for each agent $i$. Each latency and cost predictor uses Hoeffding Tree Regression, and performance predictor uses Hoeffding Tree Classifier. At decision time, we capture system load features: global router inflight/rate ($I^{(r)}, R^{(r)}$) and agent specific inflight/rate ($I_i, R_i$). We compute normalized utilization $u_i = I_i / \max(1, B_i)$. For each pair $(i, j)$, we construct a feature vector:
\vspace{-5pt}
\begin{equation}
\setlength\belowdisplayskip{5pt}
\mathbf{x}_{ij} = \bigl(|p_j|,\ t_j,\ \omega_{ij},\ I^{(r)}, R^{(r)},\ I_i, R_i, B_i, u_i,\ \xi_j\bigr),
\label{eq:predictor_features}
\end{equation}
where $t_j$ is the turn index and $\xi_j$ denotes metadata (e.g., domain tag). The Hoeffding Tree predictor outputs the estimates $(\hat{L}_{ij}, \hat{C}_{ij}, \hat{Q}_{ij}) = g_i(\mathbf{x}_{ij})$.

To reduce cold-start bias, IEMAS optionally performs a brief startup warm-up by
issuing a small number of representative multi-turn dialogues to each agent to seed
the predictors and establish initial cache state; latency labels during warm-up can be
kept conservative to avoid one-time initialization artifacts.

\textbf{Feedback Accounting:}
After execution, the proxy records the observed latency $L^{\mathrm{obs}}_{ij}$ and computes the realized cost based on token usage. Let $\pi^{\mathrm{miss}}_i$ and $\pi^{\mathrm{hit}}_i$ be the prices for uncached (miss) and cached (hit) prompt tokens $n$, respectively. The observed cost is\cite{bergemann2025economics}:
\vspace{-5pt}
\begin{align}
\setlength\belowdisplayskip{5pt}
C^{\mathrm{obs}}_{ij}
&=\pi^{\mathrm{miss}}_i (n^{\mathrm{prompt}}_{j} - n^{\mathrm{hit}}_{ij})
+\pi^{\mathrm{hit}}_i n^{\mathrm{hit}}_{ij}
+\pi^{\mathrm{out}}_i n^{\mathrm{gen}}_{j}.
\label{eq:observed_cost}
\end{align}
The predictor is updated online by the $(L^{\mathrm{obs}}_{ij}, C^{\mathrm{obs}}_{ij}, P^{\mathrm{obs}}_{ij})$, enabling the system to learn the agent's true performance characteristics and cache behavior over time.

It is important to note that prior research has established robust predictive schemes for estimating LLM performance, latency and cost \cite{sun2024llumnix, chen2025kairos,feng2025graphrouter}, as well as advanced KV orchestration techniques for fine-gained utilization, such as CacheBlend \cite{yao2025cacheblend}. Our work distinguishes itself by focusing on the \textit{mechanism design} layer rather than low-level estimation or memory virtualization. IEMAS treats these system-level metrics as inputs to guide incentive-compatible auction valuations. Consequently, our framework is modular: state-of-the-art predictive models or novel KV management techniques can be seamlessly coupled into the IEMAS proxy to further enhance auction precision and resource efficiency.

\subsection{Multi-Agent Matching Mechanisms}
\label{sec:mechanisms}

Efficient allocation relies on clients truthfully revealing their utilities to avoid market failure. Since agent service costs are transparently quantified by the proxy via Eq~\ref{eq:observed_cost}, we assume agent costs are honest and restrict our incentive analysis to the strategic behavior of clients.

To elicit truthful participation, IEMAS adopts a welfare-maximizing allocation rule with VCG-style payments. Given a set of client requests $T_j$ and a set of agents \(S\), the proxy hub constructs candidate request--agent pairs \((i,j)\) and derives scalarized valuations \(\hat v_{i,j}\) and costs \(\hat c_{i,j}\) from the predictive models described in \S\ref{sec:kv-design}. The net welfare contribution of assigning request \(j\) to agent \(j\) is denoted by
\vspace{-5pt}
\begin{equation}
\setlength\belowdisplayskip{3pt}
    w_{i,j} = \hat v_{i,j} - \hat c_{i,j},\nonumber 
\end{equation}
and pairs with \(w_{i,j}<0\) are discarded to avoid inefficient matches. The proxy hub then solves the following welfare-maximization problem:
\vspace{-5pt}
\begin{align}
\setlength\belowdisplayskip{3pt}
\label{eq:optimal}
  \max_{x} \quad & W(C) = \sum_{j\in C}\sum_{i\in S} w_{i,j}\,x_{i,j} \\
  \text{s.t.}\quad 
  & \sum_{i\in S} x_{i,j} \le 1, \quad \forall j\in C, \nonumber\\
  & \sum_{j\in C} x_{i,j} \le q_i, \quad \forall i\in S, \nonumber\\
  & x_{i,j} \in \{0,1\}, \nonumber
\end{align}
where \(q_i\) denotes the maximum number of concurrent requests agent \(i\) can serve. Problem~\eqref{eq:optimal} is a maximum-weight bipartite \(b\)-matching problem and admits a polynomial-time solution via a min-cost max-flow (MCMF) reduction. Specifically, shown in Figure~\ref{fig:IceMAS}(c), we construct a flow network with a source connected to each task node with unit capacity, task-to-agent edges of unit capacity and cost \(-w_{i,j}\), and agent-to-sink edges with capacity \(q_i\). The edge owns positive cost will be discarded to avoid inefficient matches. Any integral min-cost max-flow in this network corresponds to a feasible matching that minimizes total cost, and hence maximizes total welfare.

Given the welfare-maximizing allocation, IEMAS implements VCG (Clarke pivot) payments to ensure incentive compatibility. Let \(W(C)\) denote the optimal welfare of~\eqref{eq:optimal} and \(W(C\setminus\{j\})\) the optimal welfare when request \(j\) is removed. If client request \(j\) is matched to agent \(i\) in the optimal allocation, its VCG payment is defined as
\vspace{-5pt}
\begin{equation}
\label{eq:vcg}
\setlength\belowdisplayskip{3pt}
p_j \;=\; W(C\setminus\{j\}) \;-\; \bigl(W(C) - w_{i,j}\bigr) + c_{i,j}, 
\end{equation}
which equals the externality that request \(j\) imposes on all other participants. This mechanism is ex-post efficient and dominant-strategy incentive compatible for clients. Symmetric payments or rebates can be defined for agents when needed.

While VCG mechanisms do not, in general, guarantee budget balance—the total payments collected from clients may differ from the total compensation paid to agents—they provide a principled benchmark for truthful, welfare-optimal allocation. Moreover, although computing VCG payments naïvely requires resolving~\eqref{eq:optimal} once per request, in practice this overhead can be significantly reduced using warm-started min-cost flow solvers and incremental re-optimization within each proxy hub.

\subsection{Algorithmic properties}
\label{sec:theoretical_analysis}

In this section, we provide a formal analysis of the economic and algorithmic properties of IEMAS. Specifically, we demonstrate that the coupling of the Min-Cost Max-Flow (MCMF) algorithm with the VCG payment rule is theoretically sound. We prove that the flow-based allocation is exact—a strict prerequisite for VCG truthfulness—thereby resolving the potential conflict between computational tractability and incentive compatibility.

\textbf{Exactness of Allocation via MCMF:}
The validity of VCG relies strictly on the allocation rule being allocatively efficient (i.e., finding the global maximum of the welfare function). If the allocation algorithm were an approximation, the dominant-strategy equilibrium of VCG would collapse. We show that our flow-based formulation is exact.

Recall the welfare maximization problem defined in Eq.~\ref{eq:optimal}. Let $G=(V, E)$ be the constructed bipartite flow network. We map the welfare weights $w_{ij}$ to edge costs such that $cost_{ij} = -w_{ij}$.

\begin{theorem}[\textbf{Allocative Efficiency}]
\label{thm:optimality}
The assignment $x^*$ produced by the Min-Cost Max-Flow (MCMF) algorithm in the IEMAS flow network maximizes the total social welfare $W(\mathcal{C})$ subject to capacity constraints.
\end{theorem}

\begin{proof}
See Appendix~\ref{sec:proof_eff}.
\end{proof}

\textbf{Incentive Compatibility: }Having established that MCMF yields the optimal allocation $x^*$, we now prove that the mechanism elicits truthful reporting of client value $v_i$.

\begin{theorem}[\textbf{Dominant Strategy Incentive Compatibility for Clients}]
Assuming truthful agents, reporting the true valuation $v_i$ is a dominant strategy for every client $i \in \mathcal{C}$ under the IEMAS mechanism.
\end{theorem}

\begin{proof}
See Appendix~\ref{sec:proof_ic}.
\end{proof}

\textbf{Budget Balance.} 
While standard VCG mechanisms are not generally budget-balanced, the specific asymmetry of our market design—where agent costs are verifiable via the proxy and clients are the primary strategic actors—allows us to guarantee that the system never runs a financial deficit.

\begin{theorem}[\textbf{Weak Budget Balance}]
\label{thm:budget_balance}
The IEMAS mechanism satisfies ex-post weak budget balance. The total payment collected from clients is sufficient to cover the total service costs incurred by agents:
\vspace{-5pt}
\begin{equation}
\setlength\belowdisplayskip{3pt}
    \sum_{j \in \mathcal{C}_{matched}} p_j \ge \sum_{i \in \mathcal{S}_{active}} c_{i}.
\end{equation}
\end{theorem}

\begin{proof}
See Appendix~\ref{sec:proof_bb}.
\end{proof}

\textbf{Computational Consistency:}
While VCG theoretically requires resolving the optimization problem $|\mathcal{S}|+1$ times, the coupling with MCMF mitigates this overhead. The computation of $W(\mathcal{C} \setminus \{i\})$ is equivalent to finding a minimum-cost flow adjustment on the \textit{residual graph} $G_f$ obtained after the initial allocation $x^*$~\cite{ahuja1993network}. By reusing the dual potentials (conceptually similar to Johnson's re-weighting algorithm) from the primary solution, the marginal cost of computing payments is significantly lower than solving from scratch~\cite{hershberger2001vickrey}. Furthermore, in bipartite assignment settings, it has been established that VCG payments can often be derived directly from the optimal dual variables of the linear relaxation, rendering the re-optimization step efficient or even instantaneous~\cite{leonard1983elicitation}. This ensures the mechanism remains practical for real-time routing.

\subsection{Agentic Hub Architecture}
\label{sec:hub-design}

The system model described above faces two fundamental challenges when deployed in large-scale, open LLM agent networks. \textbf{First}, realistic deployments may involve hundreds or thousands of independently operated agents and a high volume of concurrent client requests. Performing global all-to-all prediction, auction, and scheduling—especially during the performance estimation and VCG matching phases—would incur prohibitive latency, communication overhead, and computational cost. \textbf{Second}, when agent heterogeneity is substantial, VCG-based mechanisms may violate individual rationality (IR) constraints~\cite{liu2025cast}, as implied by the Green--Laffont impossibility theorem, which precludes the simultaneous satisfaction of incentive compatibility, individual rationality, budget balance, and allocative efficiency in general settings~\cite{green1977characterization}. We demonstrate this effect in Appendix~\ref{app:cluster}.

To address these limitations, prior work has shown that \emph{clustering-based decomposition} is an effective supplement to incentive mechanisms: it bounds problem size, enables parallel intra-cluster computation, and mitigates IR conflicts by reducing heterogeneity within each market~\cite{liu2025cast}. Building on this insight, IEMAS adopts a proxy-hub architecture to support scalable and incentive-aware scheduling in heterogeneous LLM ecosystems. Specifically, agents service pods are clustered \emph{a priori} into multiple proxy hubs according to relatively static capability signals, such as model scale, domain specialization, and benchmarked performance (e.g., OpenCompass evaluations~\cite{2023opencompass}). Incoming client requests are first routed to an appropriate hub using a lightweight, coarse-grained classifier based on task domain and quality-of-service requirements. Fine-grained IEMAS routing, prediction, and VCG-based matching are then executed locally within the selected hub. This two-stage routing and allocation process substantially reduces the dimensionality of the matching problem while preserving the economic efficiency benefits of incentive-aware scheduling.

Each proxy hub acts as a mediation layer providing authentication, admission control, fine-grained accounting, and KV-cache management, thereby decoupling global market coordination from low-level inference execution. Instead of continuous real-time communication, proxy hubs periodically publish standardized, privacy-preserving metadata—such as price signals, available capacity, and compact cache-state summaries—which is sufficient for constructing efficient matchings while reducing bandwidth overhead and protecting proprietary model details. Furthermore, since the incentive mechanisms in \S\ref{sec:mechanisms} operate on bipartite request--agent graphs, proxy hubs naturally host auction execution, VCG payment computation, and cache-aware scheduling. By maintaining local KV-prefix ledger in \S\ref{sec:kv-design}, proxies enable reuse-aware matching without exposing raw prompts, model parameters, or inference traces. This separation of concerns allows IEMAS to achieve scalability, economic robustness, cache efficiency, and privacy-preserving distributed inference at Internet scale. Implementation details are deferred to Appendix~\ref{sec:implement}.

\begin{table*}[t]
 \setlength{\abovecaptionskip}{-0.3cm}
  \small
    \centering
    \setlength{\tabcolsep}{4pt} 
    \begin{tabular}{lccccccccc}
        \toprule
        \multirow{2}{*}{\textbf{Method}} & \multicolumn{3}{c}{\textbf{CoQA (Multi-turn)}} & \multicolumn{3}{c}{\textbf{QuAC (Long Context)}} & \multicolumn{3}{c}{\textbf{HotpotQA (Reasoning)}} \\
        \cmidrule(lr){2-4} \cmidrule(lr){5-7} \cmidrule(lr){8-10}
         & \textbf{KV (\%)} & \textbf{Cost} & \textbf{Lat (ms)} & \textbf{KV (\%)} & \textbf{Cost} & \textbf{Lat (ms)} & \textbf{KV (\%)} & \textbf{Cost} & \textbf{Lat (ms)} \\
        \midrule
        GraphRouter & 0.364 & 10.655 & 423.2 & 0.310 & 17.404 & 443.1 & 0.042 & 42.243 & 1470.3 \\
        GMTRouter   & 0.531 & 11.946 & 342.2 & 0.319 & 25.302 & 472.6 & 0.048 & 44.365 & 372.0 \\
        MFRouter    & 0.384 & 10.507 & 414.5 & 0.488 & 15.372 & 306.2 & 0.054 & 39.101 & 389.2 \\
        RouterDC    & 0.482 & 9.527 & 357.5 & 0.242 & 19.443 & 404.9 & 0.046 & 34.180 & 2139.8 \\
        Random      & 0.165 & 13.65 & 452.6 & 0.343 & 24.927 & 189.1 & 0.044 & 43.683 & 404.4 \\
        \midrule
        \textbf{IEMAS (Ours)} & \textbf{0.802} & \textbf{6.944} & \textbf{354.5} & \textbf{0.636} & \textbf{15.868} & \textbf{162.1} & \textbf{0.089} & \textbf{28.694} & \textbf{284.2} \\
        \bottomrule
    \end{tabular}
    \caption{Average System Efficiency Comparison across three benchmarks. \textbf{KV (\%)} denotes cache hit rate ($\uparrow$) and \textbf{Lat} is the TTFT ($\downarrow$).}
    \label{tab:welfare_expanded}
\end{table*}

\section{Experiments}
\label{sec:experiments}

We evaluate IEMAS on critical dimensions: (1) \textbf{System Efficiency}, measuring whether cache-aware routing reduces latency and overhead; (2) \textbf{Economic Robustness}, verifying that the mechanism incentivizes truthful reporting and maximizes social welfare.

\subsection{Experimental Setup}

\noindent
\textbf{Experimental Setup.}
We profiled the \texttt{vLLM} inference engine~\cite{kwon2023vllm} on heterogeneous nodes equipped with NVIDIA RTX 4090 and RTX 6000 GPUs. To faithfully simulate the resource constraints and frequent cache evictions typical of open agent networks, we fixed the concurrent query batch buffer size at 12 and restricted the vLLM \texttt{gpu\_memory\_utilization} parameter to 0.6. Our agent population is instantiated with a diverse set of models, including LLaMA-3-7B~\cite{dubey2024llama}, Qwen-4B, and Qwen-8B~\cite{yang2025qwen3}. We evaluate system performance across three distinct interaction modalities: \textbf{CoQA}~\cite{reddy2019coqa} for multi-turn dialogue preservation, \textbf{QuAC}~\cite{choi2018quac} for long-context handling, and \textbf{HotpotQA}~\cite{yang2018hotpotqa} for complex reasoning tasks. We compare IEMAS against five routing strategies:
\begin{itemize}[leftmargin=*,topsep=0em]
\setlength{\parskip}{0em}
\setlength{\itemsep}{0em}
    \item \textbf{GraphRouter \cite{feng2025graphrouter}:} a graph-based LLM router that models tasks, queries, and models as a heterogeneous graph for effect/cost estimation.
    \item \textbf{GMTRouter \cite{xie2025gmtrouter}:} a personalized router using a heterogeneous graph to capture user–query–model interaction preferences. 
    \item \textbf{MFRouter \cite{ong2025routellm}:} a matrix factorization-based router that treats routing as a recommendation.
    \item \textbf{RouterDC \cite{chen2024routerdc}:} a dual-contrastive learning router that fine-tunes a pre-trained language model to align query representations.
    \item \textbf{Random:} routes queries uniformly at random as a non-learned baseline.
\end{itemize}

\begin{figure*}[htp]
\setlength{\belowcaptionskip}{-0.4cm} 
 \setlength{\abovecaptionskip}{0cm}
    \centering
        
    \begin{minipage}{0.68\textwidth}
        \centering
        \includegraphics[width=\linewidth]{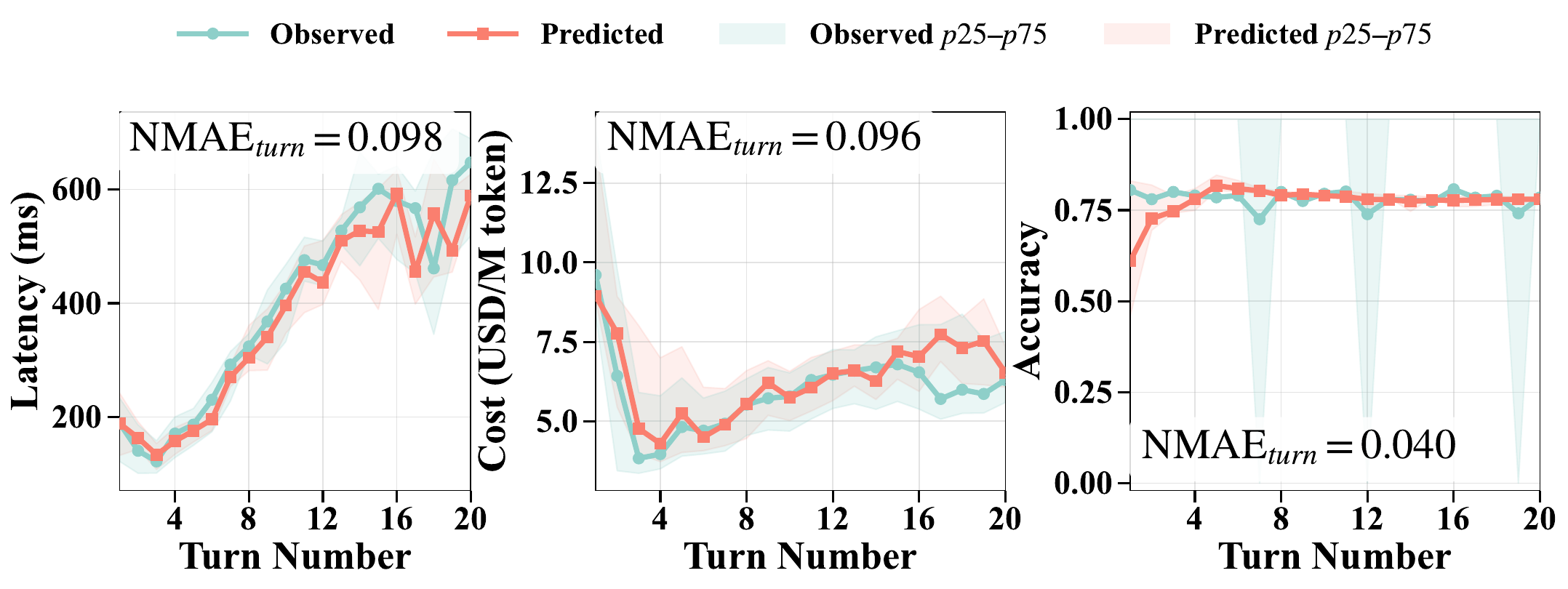}
        \caption{\small The Predictive Model for QoS Factors of CoQA dataset.}
        \label{fig:predict}
    \end{minipage}%
    \hfill
    \begin{minipage}{0.3\textwidth}
        \centering
        \includegraphics[width=\linewidth]{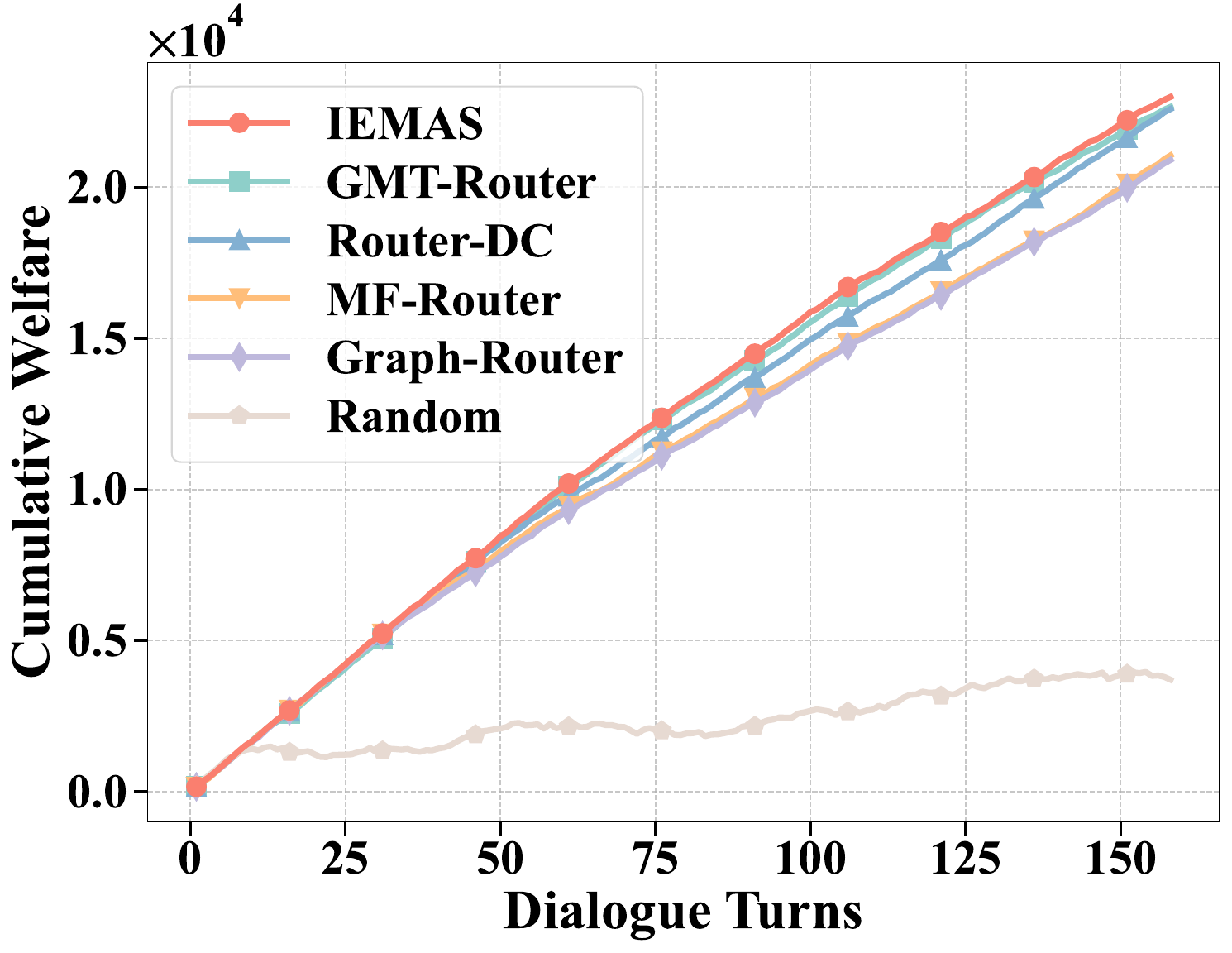}
        \caption{\small Social Welfare Comparison.}
        \label{fig:welfare}
    \end{minipage}%

\end{figure*}

\begin{figure}
\setlength{\belowcaptionskip}{-0.4cm} 
 \setlength{\abovecaptionskip}{0cm}
    \centering
    \begin{minipage}{0.23\textwidth}
        \centering
        \includegraphics[width=\linewidth]{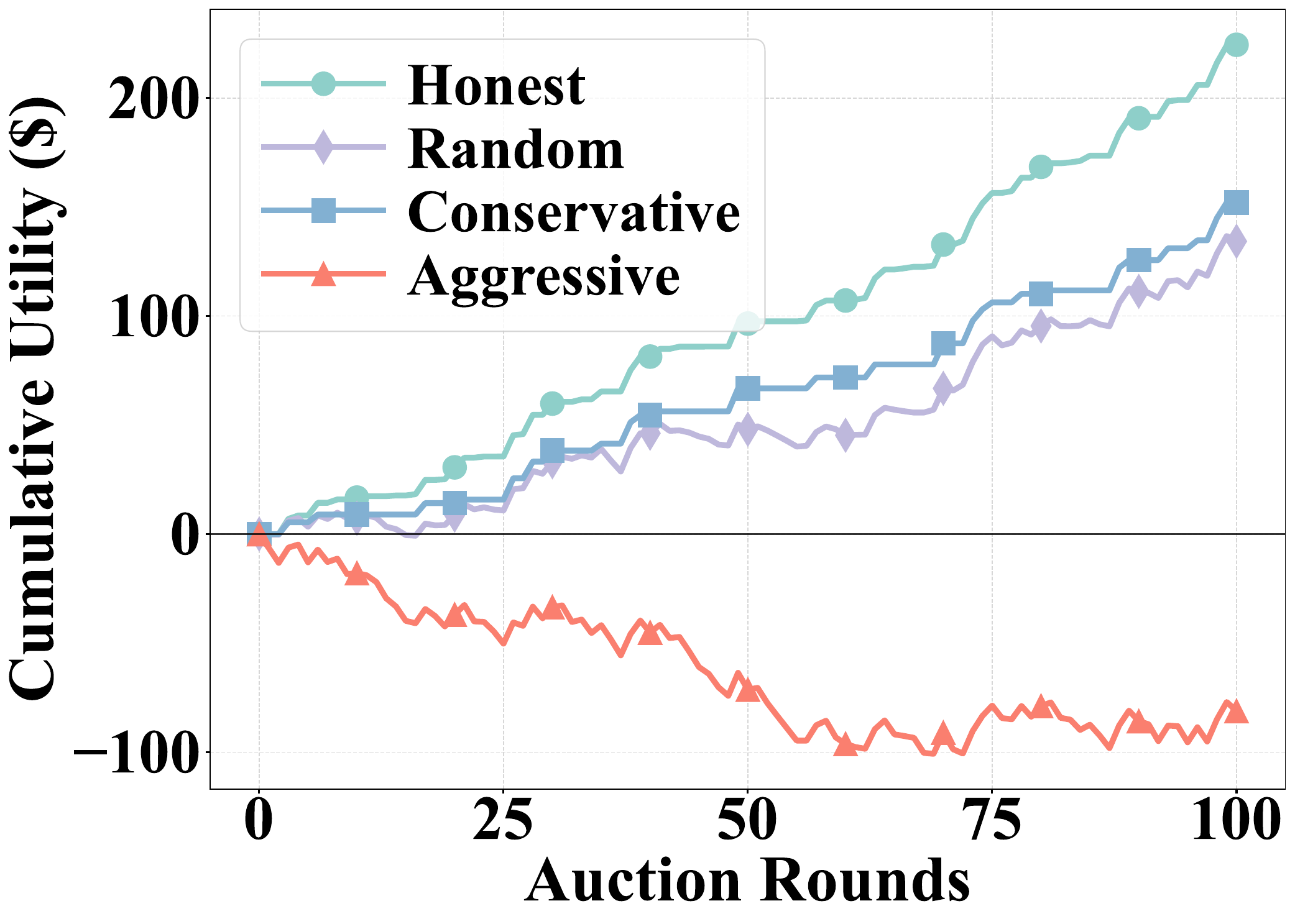}
        \caption{\small The Utility under Different Auction Strategies in VCG Auction.}
        \label{fig:vcg}
    \end{minipage}%
    \hfill
    \begin{minipage}{0.242\textwidth}
        \centering
        \includegraphics[width=\linewidth]{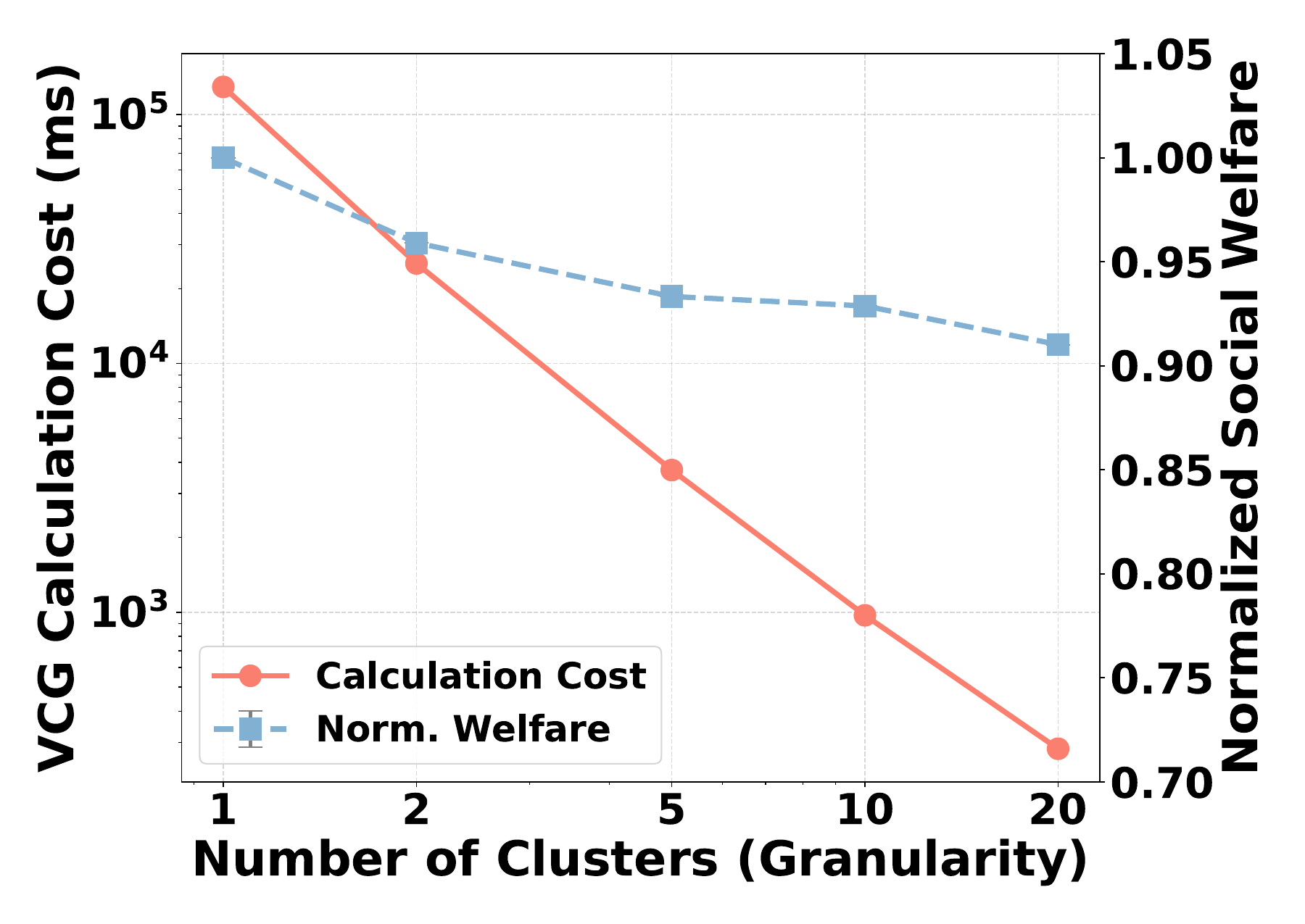}
        \caption{\small The Social Welfare and Time Cost under Different Cluster Levels in Routing. }
        \label{fig:cluser_time}
    \end{minipage}%
\end{figure}

\subsection{System Efficiency}

We evaluate resource reuse efficiency via the KV-cache hit rate, Latency, and Service Cost, as detailed in Table~\ref{tab:welfare_expanded}. 

\textbf{KV Cache Reuse.} On multi-turn datasets where context retention is critical, IEMAS demonstrates superior affinity scheduling. Specifically, on \textbf{CoQA}, IEMAS achieves a dominant cache hit rate of \textbf{80.2\%}, significantly outperforming the strongest baseline (GMTRouter at 53.1\%). Even on long-context tasks like \textbf{QuAC}, IEMAS maintains a \textbf{63.6\%} hit rate compared to the baseline average of $\sim$30-40\%. This confirms that our affinity-scoring mechanism successfully preserves context locality even in complex dialogue scenarios.

\textbf{Latency Reduction.} The impact of cache reuse on TTFT is evident in the \textit{Lat} column. On the \textbf{HotpotQA} reasoning benchmark, IEMAS reduces median latency to \textbf{284.2ms}, avoiding the severe congestion observed in RouterDC (2139.8ms) and outperforming GMTRouter (372.0ms). Similarly, on \textbf{QuAC}, IEMAS achieves the lowest latency of \textbf{162.1ms}, a \textbf{1.9$\times$} speedup over MFRouter (306.2ms), by consistently routing follow-up queries to agents holding the active KV-prefix and bypassing the prefill phase.

\textbf{Cost Efficiency.} Maximizing cache hit rates directly translates to economic savings. On CoQA, IEMAS yields an average cost of \textbf{6.944}, representing a \textbf{$\approx$35\% reduction} compared to the next most efficient method (MFRouter at 10.507) and a massive reduction compared to Random routing (13.65). While HotpotQA exhibits lower overall cache reuse (0.089) due to task nature, IEMAS still achieves the lowest cost (\textbf{28.694}) among all methods, proving that the VCG-based mechanism effectively optimizes for cost even when cache opportunities are scarce.

\subsection{Predictive Accuracy}
The efficiency of the auction mechanism hinges on the Proxy Hub's ability to accurately calibrate agent Latency ($L$), Cost ($C$), and Accuracy ($P$). As illustrated in Figure~\ref{fig:predict}, our online Hoeffding Tree predictor demonstrates robust predictive capability across multi-turn interactions. The figure compares observed versus predicted values over 20 turns, revealing tight alignment between the model's estimates and ground truth. Specifically, the model achieves a low Normalized Mean Absolute Error (NMAE) of \textbf{0.101} for latency and \textbf{0.090} for cost, effectively capturing the non-linear increase in resource consumption as context length grows. Furthermore, accuracy prediction remains stable with an NMAE of \textbf{0.069}, acting as a reliable expected value filter against the stochastic noise of observed performance, ensuring the auction mechanism operates on calibrated risk bounds.

\subsection{Economic Analysis: Truthfulness and Welfare}

A core contribution of IEMAS is Incentive Compatibility (IC). We validate this by introducing strategic agents who attempt to game the system.

\textbf{Truthfulness Validation.} We simulate four agent bidding strategy: Honest, aggressive, conservative and random who report truthful, consistly high/low and random based on the true value. Figure~\ref{fig:vcg} illustrates the cumulative utility (profit) over 100 auction rounds. Under the VCG-based payment rule, truthful agents see steady profit growth. In contrast, the other strategic agents incur penalties due to performance deviations, resulting in  strictly low profit at any round. This empirically demonstrates that truth-telling is a dominant strategy in IEMAS.

\textbf{Clustering Trade-off Analysis.} To quantify the balance between computational scalability and allocative efficiency, we conducted a sensitivity analysis by varying the number of proxy hub clusters $K$. We maintained a constant global environment of $M=100$ agents and $N=200$ concurrent tasks. As illustrated in Figure~\ref{fig:cluser_time}, increasing $K$ leads to a sharp reduction in solver latency, as the complexity of the Min-Cost Max-Flow (MCMF) algorithm scales super-linearly with problem size. Crucially, this performance gain incurs only a marginal degradation in global social welfare. The results confirm that domain-based clustering effectively partitions the market, significantly reducing scheduling overhead while maintaining near-optimal allocation efficiency. We further analysis the effect of different clustering schemes in Appendix~\ref{app:cluster}.

\textbf{Social Welfare.} We define Social Welfare as the cumulative sum of client utility minus agent costs over the session duration. Figure~\ref{fig:welfare} illustrates the welfare accumulation across 160 dialogue turns. \textbf{IEMAS} consistently maintains the steepest growth trajectory, demonstrating superior long-term efficiency. While state-of-the-art baselines like \textit{GMT-Router} and \textit{Router-DC} track closely, IEMAS maintains a sustained lead, exceeding all baselines by the final turn. This gap highlights the advantage of our incentive-compatible mechanism: by explicitly pricing cache affinity in the VCG valuation, IEMAS minimizes the cost of context that degrades the net utility of other routing strategies over long horizons. In contrast, the \textbf{Random} baseline fails to generate meaningful welfare, validating the necessity of intelligent coordination.

\vspace{-5pt}
\section{Conclusion}
The transition to an open \textbf{Internet of Agents} requires reconciling the strategic autonomy of decentralized providers with the physical realities of efficient LLM inference. In this work, we introduced \textbf{IEMAS}, a routing framework that bridges this gap by treating the KV-cache not merely as a system buffer, but as a pricable economic asset. By co-designing a resource-aware predictive model with a truthful VCG mechanism, IEMAS solves the fundamental tension between maximizing global social welfare and ensuring individual rationality. Our theoretical analysis confirms that the mechanism guarantees truthful reporting and weak budget balance. Empirically, extensive simulations demonstrate that this alignment yields tangible performance gains: IEMAS achieves an \textbf{80\%} KV-cache hit rate, reduces average service cost by \textbf{35\%}, and lowers end-to-end latency by \textbf{2.9$\times$} compared to state-of-the-art baselines. These results suggest that rigorous mechanism design, when tightly coupled with system-level observables, is a prerequisite for scaling the future Agentic Web.

\section*{Impact Statement}

This paper presents work whose goal is to advance the field of 
Machine Learning. There are many potential societal consequences 
of our work, none which we feel must be specifically highlighted here.

\bibliography{example_paper,ref,related}
\bibliographystyle{icml2026}

\newpage
\appendix
\onecolumn

\section{Proof}

\subsection{Efficiency}
\label{sec:proof_eff}

The MCMF algorithm minimizes the total cost function on the network:
\vspace{-5pt}
\begin{equation}
\setlength\belowdisplayskip{3pt}
    \min_{f} \sum_{(u,v) \in E} cost_{uv} \cdot f_{uv} \iff \max_{x} \sum_{i \in \mathcal{S}} \sum_{j \in \mathcal{C}} w_{ij} x_{ij}
\end{equation}
The constraint matrix $A$ of this bipartite matching problem corresponds to the node-edge incidence matrix. The rows represent flow conservation constraints at Task nodes $j$ (capacity 1) and Agent nodes $i$ (capacity $B_i$). By the \textit{Hoffman-Kruskal Theorem}~\cite{hoffman1956integral}, the incidence matrix of a bipartite graph is \textit{Totally Unimodular} (TU). A fundamental property of TU matrices is that for any integral capacity vector $b$ (here, $1$ and $B_i$ are integers), the polyhedron of feasible solutions $P = \{x | Ax \le b, x \ge 0\}$ has integral vertices~\cite{schrijver2003combinatorial}.
Consequently, the solution found by the MCMF algorithm—which solves the linear relaxation of the problem—is guaranteed to be integral ($x_{ij}^* \in \{0, 1\}$). Since MCMF provides an exact solution to the minimum cost circulation problem, the resulting allocation $x^*$ is the global maximizer of social welfare.

\subsection{Truthfulness}
\label{sec:proof_ic}

Let $u_j(\hat{v}_j, v_{-j})$ denote the utility of client $i$ when reporting valuation $\hat{v}_j$, while other clients report $v_{-j}$. The utility is defined as the realized value minus the payment:
\vspace{-5pt}
\begin{equation}
\setlength\belowdisplayskip{3pt}
    u_j = v_j - p_j
\end{equation}
Substitute the VCG payment rule from Eq.~\eqref{eq:vcg} into the utility function. Note that $W(\mathcal{C})$ depends on the reported $\hat{v}_j$, so we expand $W(\mathcal{C}) = (\hat{v}_j - c_{i,j}) + \sum_{k \neq j} w_{k, s_k}^*$, where the sum represents the optimal welfare of all other matched pairs in the presence of $j$.
\vspace{-5pt}
\begin{align}
\setlength\belowdisplayskip{3pt}
    u_j &= v_j - \left[ W(\mathcal{C}\setminus\{j\}) - (W(\mathcal{C}) - (\hat{v}_j - c_{i,j})) + c_{i,j} \right] \nonumber \\
    &= v_j - W(\mathcal{C}\setminus\{j\}) + W(\mathcal{C}) - \hat{v}_j + c_{i,j} - c_{i,j} \nonumber \\
    &= (v_j- \hat{v}_j) + W(\mathcal{C}) - W(\mathcal{C}\setminus\{j\})
\end{align}
If the client reports truthfully such that $\hat{v}_j = v_j$, the term $(v_j - \hat{v}_j)$ vanishes:
\vspace{-5pt}
\begin{equation}
\setlength\belowdisplayskip{3pt}
    u_j = \underbrace{W(\mathcal{C})}_{\text{Total Social Welfare}} - \underbrace{W(\mathcal{C} \setminus \{j\})}_{\text{Constant independent of } j}
\end{equation}
The term $W(\mathcal{C} \setminus \{j\})$ represents the system welfare if client $j$ were absent, which is independent of client $j$'s strategy. Thus, maximizing individual utility $u_j$ is mathematically equivalent to maximizing the total social welfare $W(\mathcal{C})$. 
Since the MCMF algorithm finds the exact allocation that maximizes $W(\mathcal{C})$ based on reported values, the client maximizes their own utility $u_j$ if and only if they provide the true valuation $v_j$ that allows the mechanism to optimize the true social welfare. Under-reporting or over-reporting $v_j$ may lead to a suboptimal allocation (or no allocation) that yields a strictly lower utility.

\subsection{Budget Balance}
\label{sec:proof_bb}

Consider a single matched transaction between client \(j\) and agent \(i\). The platform collects \(p_j\) from the client and compensates the agent \(c_{i,j}\). The platform's net surplus for this transaction is:
\begin{align}
    \Delta_j &= p_j - c_{i,j} \nonumber \\
    \text{Using Eq.~\eqref{eq:vcg}:} \quad \Delta_j &= \left[ W(\mathcal{C}\setminus\{j\}) - (W(\mathcal{C}) - w_{i,j}) + c_{i,j} \right] - c_{i,j} \nonumber \\
    \Delta_j &= W(\mathcal{C}\setminus\{j\}) - \bigl(W(\mathcal{C}) - w_{i,j}\bigr)
\end{align}
Here, \(W(\mathcal{C}\setminus\{j\})\) represents the maximum social welfare achievable by the remaining participants \emph{without} client \(j\). The term \((W(\mathcal{C}) - w_{i,j})\) represents the actual welfare of the remaining participants in the optimal allocation \emph{with} client \(j\) present.
Since agent capacity is finite (resource contention), the presence of client \(j\) can only decrease or leave unchanged the aggregate welfare available to others (by consuming a slot that might have gone to another high-value task). It cannot increase the welfare of others. Therefore:
\begin{equation}
    W(\mathcal{C}\setminus\{j\}) \ge W(\mathcal{C}) - w_{i,j} \implies \Delta_j \ge 0.
\end{equation}
Since the surplus \(\Delta_j\) is non-negative for every matched request, the sum over all requests is non-negative.

\noindent
\textit{Remark:} The non-negative surplus \(\sum \Delta_i\) represents the "information rent" or externality tax collected by the protocol. In the IEMAS implementation, this surplus can be retained by the Hub as a maintenance fee or redistributed to agents as a long-term participation incentive (e.g., staking rewards) without violating the budget constraint.

\section{Extended Analysis}

\subsection{Cluster Strategy Effect Analysis}
\label{app:cluster}

As discussed in Section 4.2, IEMAS employs the VCG mechanism to ensure incentive compatibility. While theoretically robust, the computational cost of calculating Clarke pivot payments is substantial. Specifically, calculating the externality for each request requires resolving the Min-Cost Max-Flow (MCMF) problem $N$ times for $N$ concurrent requests. In a global, flat market with thousands of agents, this $O(N \cdot T_{MCMF})$ complexity becomes prohibitive for real-time inference routing.

To address this, Section 4.4 introduces a hierarchical Proxy Hub architecture. While clustering agents reduces the search space—potentially lowering the theoretical global maximum social welfare (SW)—our analysis shows that the specific characteristics of LLM inference mitigate this loss.

\textbf{Market Fragmentation vs. Cache Locality.} 
In generic resource allocation, partitioning agents into disjoint clusters restricts the solution space, creating a "Fragmentation Penalty" where a task cannot access an idle agent in a different cluster. However, LLM inference efficiency is dominated by KV-cache locality. Global routing often suffers from context thrashing, where semantically similar queries are dispersed. By clustering agents based on domain specialization (as verified in Appendix A.1), the system maximizes intra-cluster cache reuse ($o_{ij}$). This gain in execution efficiency (lower cost $C_i$ and latency $L_j$) effectively offsets the loss in matching optionality.

\begin{figure}[h!]
    \centering
    \includegraphics[width=0.6\textwidth]{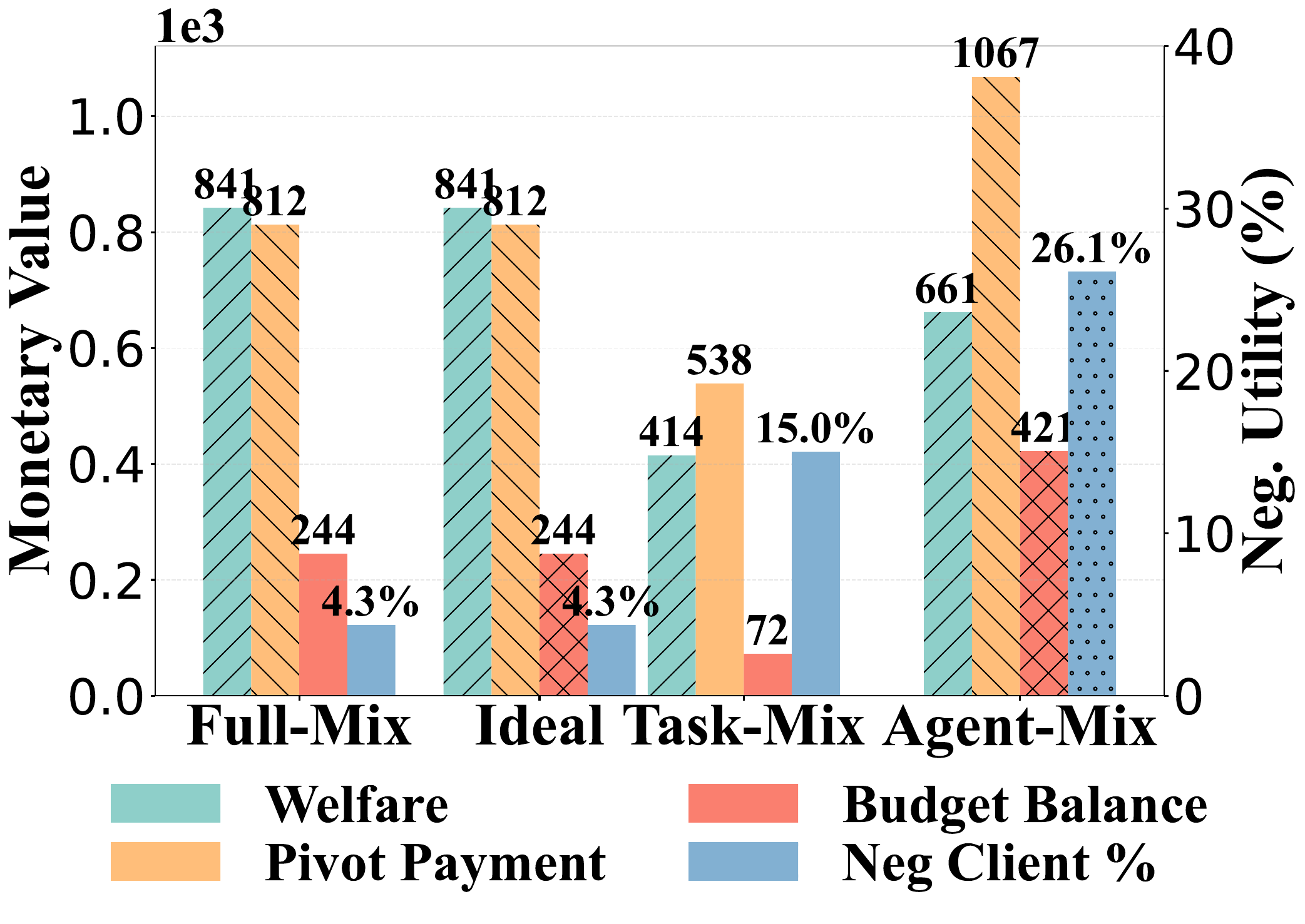} 
    \caption{The Economics Performance under Different Cluster Schemes.}
    \label{fig:cluster}
\end{figure}

\emph{Full-Mix} is a heterogeneous baseline with no prior alignment between tasks and agents. \emph{Ideal} represents homogeneous alignment, where tasks and agents are pre-clustered to minimize matching entropy. \emph{Task-Mix} clusters agents by specialization while tasks remain heterogeneous, whereas \emph{Agent-Mix} clusters tasks while agents remain heterogeneous, assessing how each side absorbs structural mismatch.

Figure~\ref{fig:cluster} shows that reasonable cluster can cause ideal situation, which can maintain high social benefits and achieve similar results as Full-Mix. However one-sided centric clustering method such as Task-Mix or Agent-Mix can cause task congestion and cutthroat competition as high-level agent service resource are relatively limited, which can reduce social welfare, increase negative client and influence IR.

\section{System Integration and Implement}
\label{sec:implement}

\begin{figure}[htbp]
    \centering
    \includegraphics[width=1\textwidth]{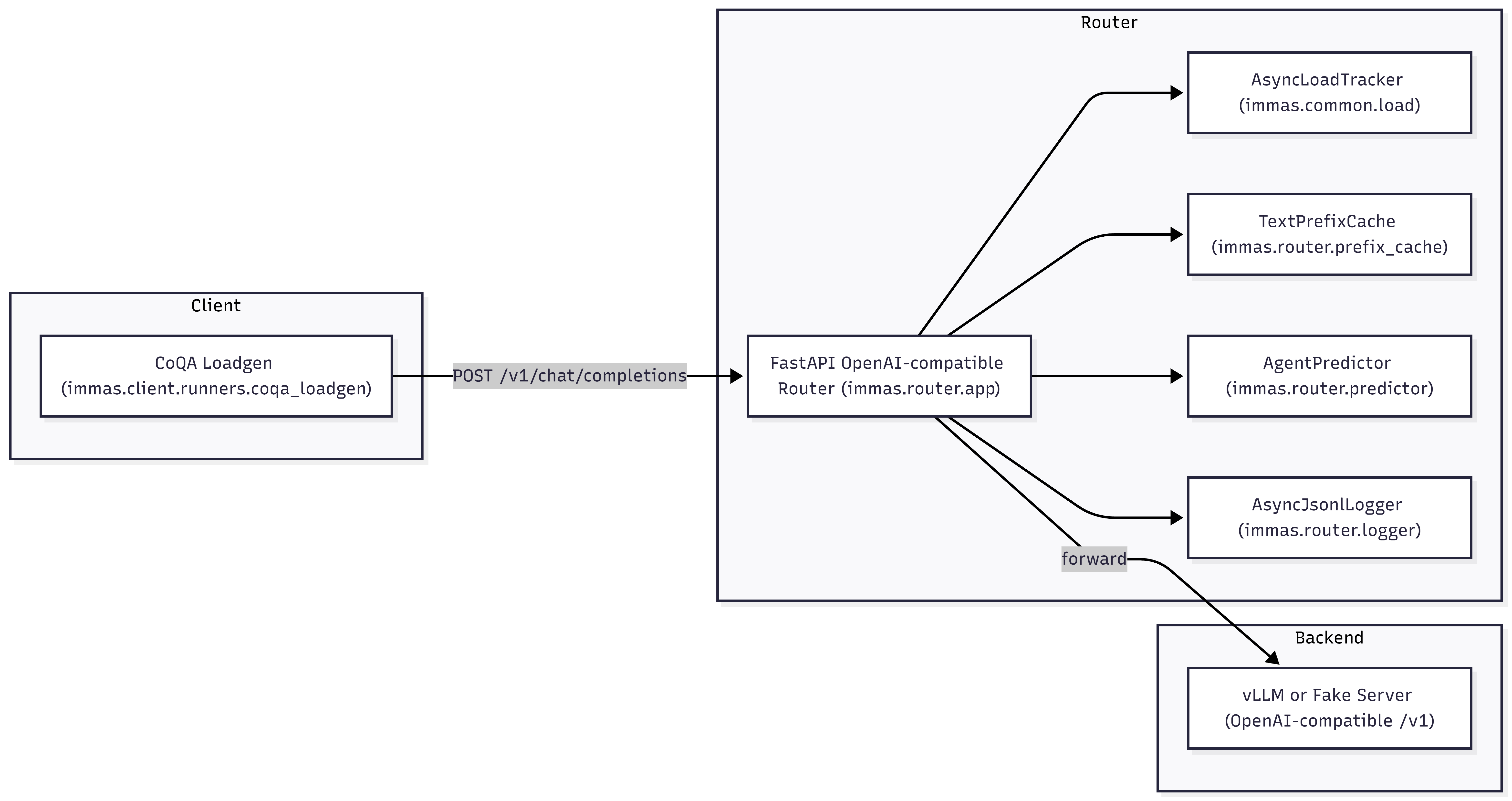} 
    \caption{The System Structure of IEMAS}
    \label{fig:IMMASsystem}
\end{figure}

This appendix details the software architecture and implementation logic of IEMAS. The system is designed as an asynchronous, event-driven proxy service that acts as an intelligent router between a high-throughput client load generator and a heterogeneous pool of Large Language Model (LLM) backends. The implementation relies primarily on Python's \texttt{asyncio} for concurrency and \texttt{FastAPI} for the web service layer.

The architecture is divided into three primary functional blocks: the \textbf{Client} (load generation), the \textbf{Router} (state management, prediction, and allocation), and the \textbf{Backend Interface} (protocol translation and telemetry).

\subsection{Client Implementation}
\label{app:client_impl}

The client module (\texttt{iemas.client}) implements a pure load generator designed to stress-test the routing infrastructure using the CoQA dataset.

\textbf{State Management and Concurrency.}
The client executes multiple dialogues concurrently using \texttt{asyncio.gather}. Within each dialogue task (\texttt{run\_dialogue}), turns are executed sequentially to preserve conversational causality---Turn $N$ must complete before Turn $N+1$ begins. The client maintains a local conversation history list, appending user questions and assistant answers verbatim as they occur. Concurrency is controlled via a global \texttt{asyncio.Semaphore} (\texttt{max\_concurrency}), ensuring a fixed upper bound on the number of simultaneous active requests.

\textbf{Protocol and Tracing.}
Unlike standard benchmarks that might maintain a persistent connection, the client initiates a fresh HTTP POST request for every turn to the router's \texttt{/v1/chat/completions} endpoint. To enable the router to track sessions across these stateless HTTP requests, the client injects custom tracing headers:
\begin{itemize}
    \item \texttt{X-IEMAS-RUN-ID}: A unique identifier for the experiment run.
    \item \texttt{X-IEMAS-DIALOGUE-ID}: The unique ID of the current conversation.
    \item \texttt{X-IEMAS-TURN-NUMBER}: The monotonic turn index (1-based).
    \item \texttt{X-IEMAS-SOURCE}: The provenance of the dialogue text.
\end{itemize}

\subsection{Router Architecture}
\label{app:router_arch}

The router (\texttt{iemas.router}) acts as the central decision-making entity. It does not perform inference itself but routes requests to backend endpoints defined in a YAML configuration file.

\subsubsection{Asynchronous Micro-Batching}
To enable collective decision-making (such as auctions) rather than greedy per-request routing, the router implements a \texttt{MicroBatcher} component. Incoming requests are not processed immediately; they are wrapped in a \texttt{PendingChatCompletion} object containing the request body and an unresolved \texttt{asyncio.Future}.

These pending objects are submitted to the \texttt{MicroBatcher}, which buffers them into an \texttt{asyncio.Queue}. The batcher yields a batch for processing when one of two conditions is met:
\begin{enumerate}
    \item \textbf{Size Threshold:} The queue size reaches \texttt{max\_batch\_size} (e.g., 16 requests).
    \item \textbf{Time Threshold:} The oldest request in the queue has waited for \texttt{max\_wait\_ms} (e.g., 10ms).
\end{enumerate}
This mechanism ensures that the overhead of batching remains bounded by a tight latency budget. The batcher runs as a background task (\texttt{\_run\_loop}) and invokes a callback handler (\texttt{handle\_chat\_batch}) for every emitted batch.

\subsubsection{Prefix Caching and State Tracking}
Effective routing requires estimating the state of the Key-Value (KV) cache at each backend without direct access to backend memory. The router maintains a local \texttt{TextPrefixCache}.

\textbf{Data Structure.} The cache is a dictionary mapping the tuple \texttt{(backend\_id, model, dialogue\_id)} to the full text of the prompt used in the \textit{previous} turn.

\textbf{Feature Extraction.} For every request in a micro-batch, the router computes the Longest Common Prefix (LCP) between the current request's serialized prompt and the cached text for each candidate backend. This computation yields a \texttt{ratio} (LCP length / prompt length) and \texttt{lcp\_chars}, which serve as proxy features for the backend's actual KV cache hit rate.

\textbf{Eviction Heuristic.} The router implements a heuristic eviction policy (\texttt{should\_evict\_router\_prefix\_cache}). If a backend reports zero cached tokens in its usage statistics despite a high router-side prefix match, the router infers that the backend has evicted the KV cache and invalidates its own local record to resync state.

\subsubsection{Online Prediction Subsystem}
The system employs an \texttt{AsyncBackendPredictorPool} to manage online learning models. To prevent interference between the learning processes of different backends, the system maintains a separate, independent \texttt{AgentPredictor} instance for each backend.

\textbf{Feature Engineering.} The \texttt{PredictorInput} dataclass aggregates features at routing time. These include:
\begin{itemize}
    \item \textbf{Request Features:} Prompt length (characters), turn number.
    \item \textbf{Cache Features:} The computed \texttt{kvmatch\_text} ratio.
    \item \textbf{System Load:} Router-wide inflight requests and global Requests Per Second (RPS).
    \item \textbf{Local Load:} Backend-specific metrics including \texttt{backend\_inflight}, \texttt{backend\_rps}, and utilization (inflight divided by capacity).
\end{itemize}

\textbf{Model Updates.} The predictors utilize Hoeffding Trees (\texttt{HoeffdingTreeRegressor}, \texttt{HoeffdingTreeClassifier}) from the \texttt{river} library. They support \texttt{partial\_fit} via the \texttt{learn\_one} method, allowing the system to update the model incrementally immediately after a request completes.

\subsubsection{Auction-Based Routing Logic}
When the routing policy is set to \texttt{auction}, the router executes the \texttt{select\_backends\_auction} function.

\textbf{Graph Construction.} The allocation problem is modeled as a Min-Cost Max-Flow (MCMF) network flow graph:
\begin{itemize}
    \item A \textbf{Source} node connects to all Request nodes with capacity 1 and cost 0.
    \item \textbf{Request} nodes connect to \textbf{Backend} nodes. An edge exists only if the calculated welfare for that assignment is positive. The capacity is 1, and the cost is set to negative welfare (since standard solvers minimize cost).
    \item \textbf{Backend} nodes connect to a \textbf{Sink} node with capacity equal to the backend's available concurrency slots.
\end{itemize}

\textbf{Solver.} The router includes a custom, dependency-free implementation of the Successive Shortest Path algorithm (\texttt{iemas.router.auction.mcmf}). It uses Bellman-Ford potentials to handle negative edge costs and Dijkstra's algorithm for finding augmenting paths.

\textbf{VCG Payment Calculation.} To ensure economic robustness, the system calculates Vickrey-Clarke-Groves (VCG) payments. This requires calculating the "externality" a request imposes on others. For a batch of size $N$, the router runs the MCMF solver $N+1$ times: once to find the optimal allocation, and once for each request $i$ (removing $i$ from the graph) to calculate the counterfactual welfare.

\subsubsection{Performance Evaluation}
To provide a ground-truth signal for the performance predictor, the router evaluates response correctness asynchronously using the \texttt{PerformanceEvaluator} protocol. Two implementations are provided:
\begin{itemize}
    \item \textbf{RougeCoqaEvaluator}: Computes ROUGE-1/2/L F1 scores against the dataset gold answer.
    \item \textbf{TokenSpanCoqaEvaluator}: A deterministic evaluator that normalizes the text (lowercasing, number normalization) and checks if the gold answer tokens appear as a contiguous subsequence in the model output.
\end{itemize}

\subsection{Backend Interface and Deployment}
\label{app:backend_implementation}

The IEMAS router abstracts the underlying inference engines through a unified protocol layer defined in \texttt{iemas.router.components.backend}. This layer manages connection pooling, protocol translation, and high-fidelity telemetry injection. While the system is agnostic to the specific serving framework, our reference implementation deploys high-throughput \texttt{vLLM} instances.

\subsubsection{The \texttt{HttpOpenAIBackend} Abstraction}
The core component interacting with model servers is the \texttt{HttpOpenAIBackend} class. This class wraps an asynchronous \texttt{httpx.AsyncClient} to manage persistent HTTP/1.1 connections to the backend's \texttt{/v1/chat/completions} endpoint.

\textbf{Authentication and Configuration.}
Each backend instance is initialized with a \texttt{base\_url\_v1} and an optional \texttt{api\_key}. The backend automatically injects the \texttt{Authorization: Bearer <key>} header into every outbound request. The configuration also supports distinct pricing models for input, cached-input, and output tokens, which are normalized into a unified cost proxy during initialization.

\subsubsection{Telemetry Injection via Forced Streaming}
A critical requirement for the router is to measure Time-To-First-Token (TTFT) latency to train its prediction models, even when the client requests a standard non-streaming JSON response. To achieve this without modifying the standard OpenAI protocol, we implement a side-channel measurement technique within \texttt{forward\_chat\_completions}.

\textbf{Protocol Translation Logic:}
When the router processes a request, it modifies the payload before forwarding it to the backend:
\begin{enumerate}
    \item \textbf{Forced Streaming:} The router explicitly sets \texttt{stream=True} in the request body, regardless of the client's original preference.
    \item \textbf{Stream Options:} It injects \texttt{stream\_options=\{"include\_usage": True\}} to ensure that the backend sends token usage statistics (prompt processing and completion tokens) at the end of the stream.
\end{enumerate}

\textbf{Stream Consumption and Measurement:}
The backend consumes the resulting Server-Sent Events (SSE) stream using an asynchronous iterator.
\begin{itemize}
    \item \textbf{TTFT Measurement:} The router records the monotonic timestamp ($T_{start}$) immediately before sending the request. It records a second timestamp ($T_{first}$) upon receiving the first valid SSE data chunk containing a content delta.
    \item \textbf{Response Reconstruction:} Because the client may expect a single JSON object, the router employs a \texttt{\_StreamReconstruction} class. This component aggregates the incoming \texttt{choices[].delta.content} fragments into a buffered \texttt{\_ChoiceAccum} object.
    \item \textbf{Final Assembly:} Once the stream terminates (receiving \texttt{[DONE]}), the reconstructor builds a standard \texttt{ChatCompletion} JSON object, populating the \texttt{usage} field from the final stream chunk.
\end{itemize}

\textbf{Internal Telemetry Transport:}
To pass the measured TTFT up to the router's logging and prediction layers without breaking the API contract, the system injects the measurement into the reconstructed JSON payload using a private key, \texttt{\_iemas\_t\_first\_token\_monotonic}. This key is extracted and removed by the routing pipeline before the final response is returned to the client.

\subsubsection{Reference Deployment Configuration}
The system is validated against \texttt{vLLM} backends. The deployment scripts utilize specific flags to enable the advanced memory features relied upon by the router's prefix cache and cost models.

\textbf{vLLM Configuration Flags:}
As seen in the deployment scripts (e.g., `llama.sh`, `qwen.sh`), the backends are launched with the following critical arguments:
\begin{itemize}
    \item \texttt{--enable-prefix-caching}: This enables the backend's internal BlockManager to reuse KV-cache pages across requests sharing common prefixes. This is the physical mechanism that the router's \texttt{TextPrefixCache} attempts to model.
    \item \texttt{--enable-prompt-tokens-details}: This instructs vLLM to return detailed breakdown of token usage, specifically the \texttt{cached\_tokens} field. The router uses this ground-truth signal to validate its prefix cache predictions and detect eviction events.
    \item \texttt{--max-model-len 4096}: Ensures a consistent context window across heterogeneous models to prevent out-of-memory errors during high-concurrency micro-batching.
\end{itemize}

\newpage

\section{Pseudo-code of IEMAS}

\begin{algorithm}[h]
   \caption{IEMAS: Incentive-Efficiency Routing Framework}
   \label{alg:iemas_overview}
\begin{algorithmic}[1]
   \STATE {\bfseries Input:} Batch of client tasks $\mathcal{C}$, Set of agents $\mathcal{S}$ with capacities $\{B_i\}$
   \STATE {\bfseries Global State:} Predictive Model $\mathcal{M}$, Agent Prefix Ledgers $\mathcal{L} = \{\bar{p}_{i,d}\}$
   
   \STATE \textit{// Phase 1: Cache-Aware Prediction \& Valuation}
   \FOR{each task $j \in \mathcal{C}$ \AND each agent $i \in \mathcal{S}$}
       \STATE Retrieve last cached prompt $\bar{p}_{i,d(j)}$ from $\mathcal{L}$
       \STATE Compute Cache Affinity $\omega_{ij}$ via Eq.~\ref{eq:cache_affinity}
       \STATE Extract feature vector $\mathbf{x}_{ij}$ (load, cache, metadata)
       \STATE Predict QoS: $(\hat{L}_{ij}, \hat{P}_{ij}, \hat{C}_{ij}) \leftarrow \mathcal{M}.\text{predict}(\mathbf{x}_{ij})$
       \STATE Calculate Client Valuation $v_j$ via Eq.~\ref{eq:individual_value}
       \STATE Compute Net Welfare weight $w_{ij} \leftarrow v_j - \hat{C}_{ij}$
       \IF{$w_{ij} < 0$}
           \STATE Prune edge $(i,j)$ \COMMENT{Exclude inefficient matches}
       \ENDIF
   \ENDFOR

   \STATE \textit{// Phase 2: Welfare Maximization (MCMF)}
   \STATE Construct bipartite graph $G$ with edge costs $-w_{ij}$
   \STATE Compute allocation $x^* \leftarrow \text{MinCostMaxFlow}(G)$ \COMMENT{Theorem 4.1}
   \STATE Calculate Total Welfare $W(\mathcal{C})$

   \STATE \textit{// Phase 3: VCG Payment \& Dispatch}
   \FOR{each task $j$ matched to agent $i$ in $x^*$}
       \STATE \textit{// Calculate Opportunity Cost (Externality)}
       \STATE Re-solve MCMF for $G \setminus \{j\}$ to get $W(\mathcal{C} \setminus \{j\})$
       \STATE Compute Payment $p_j$ via Eq.~\ref{eq:vcg}
       \STATE \textbf{Dispatch} task $j$ to agent $i$
   \ENDFOR

   \STATE \textit{// Phase 4: Execution \& Online Learning}
   \FOR{each completed task $j$ by agent $i$}
       \STATE Observe realized metrics $L^{\text{obs}}_{ij}, C^{\text{obs}}_{ij}, P^{\text{obs}}_{ij}$
       \STATE $\mathcal{M}.\text{update}(\mathbf{x}_{ij}, \text{labels})$ \COMMENT{Hoeffding Tree Regressor Update}
       \STATE $\mathcal{L}.\text{update}(i, d(j), \text{new\_prompt})$ \COMMENT{Update Prefix Ledger}
   \ENDFOR
\end{algorithmic}
\end{algorithm}

\end{document}